\documentclass[copyright,creativecommons]{eptcs}
\providecommand{\event}{PLACES 2016} % Name of the event you are submitting to
\usepackage{breakurl}             % Not needed if you use pdflatex only.

\usepackage{mathpartir}
\usepackage{amsmath}
\usepackage{amsthm}
\usepackage{amssymb}
\usepackage{mathtools}
\usepackage{stmaryrd}
\usepackage{xcolor}
\usepackage{enumerate,xspace}

\newtheorem{lemma}{Lemma}

\theoremstyle{definition}
\newtheorem{definition}{Definition}
\newtheorem{remark}{Remark}

\usepackage{macro}

\title{Reversible Sessions Using Monitors\thanks{Research partly supported by EU COST Actions IC1201 and IC1405.}} %\thanks{Preliminary title; Draft of \today.}}
\author{Claudio A. Mezzina
\institute{IMT School for Advanced Studies Lucca, Italy}
\and
Jorge A. P\'{e}rez
\institute{University of Groningen, The Netherlands}
}
\def\titlerunning{Reversible Sessions Using Monitors}
\def\authorrunning{C. A. Mezzina \& J. A. P\'{e}rez}
\begin{document}
\maketitle

\begin{abstract}
Much research has studied foundations for correct  and reliable
\emph{com\-muni\-ca\-tion-centric systems}. A salient approach to correctness uses %static verification based on 
\emph{session types} to enforce structured communications; a recent approach to reliability uses 
\emph{reversible} actions as a way of reacting to unanticipated  events or failures.
This note develops a simple observation: 
the  machinery  required to define asynchronous semantics and monitoring 
can also support reversible protocols. 
%As communication-centric systems are increasingly developed using heterogeneous components whose correctness cannot always be certified, system validation based on static verification falls short. As a first step towards validation techniques
%based on runtime verification, 
We propose a process framework of session communication in which monitors support reversibility. 
A  key novelty in our approach are \emph{session types with present and past}, which allow us to streamline 
the semantics of reversible actions.
\end{abstract}

\section{Introduction}
Much research has studied foundations  
for reliable
\emph{com\-muni\-ca\-tion-centric} software systems.
Our interest is in process frameworks that, building on core calculi for concurrency, offer analysis techniques for message-passing programs. While early frameworks focused on (static) verification of protocol correctness, as enforced by 
properties such as 
safety, fidelity, and progress (deadlock-freedom)~(see, e.g.,~\cite{DBLP:conf/esop/HondaVK98,DBLP:conf/popl/HondaYC08,DBLP:journals/tcs/CairesV10}),  extensions with external mechanisms (such as, e.g., exceptions, interruptions, and compensations~\cite{DBLP:conf/tgc/CairesFV08,DBLP:journals/fmsd/DemangeonHHNY15,DBLP:journals/mscs/CapecchiGY16}, adaptation~\cite{DBLP:journals/scp/GiustoP15}, monitoring~\cite{LGP2016}, reversibility~\cite{DBLP:journals/jlp/TiezziY15}) have been proposed to enforce protocol correctness 
even in the presence of unanticipated events (say, failures or new requirements).

Comprehensive approaches to correctness and reliability, which address and enforce both kinds of requirements, seem indispensable in the principled design of communication-centric systems.
As these systems are increasingly built using heterogeneous services whose provenance/correctness cannot always be certified in advance, static validation techniques (such as type systems) fall short. Correctness must then be guaranteed by mechanisms for reliability,  which may inspect the (visible) behavior of interacting services and take action if they deviate from prescribed communication protocols.

We report on ongoing work aimed at uniform approaches to correct, reliable communicating systems. %, in the above sense.
We address the interplay between \emph{session types} and models of \emph{reversible computation}: models of concurrency in which the usual \emph{forward} semantics is coupled with a \emph{backward} semantics that allows one to ``undo'' process actions~\cite{DanosK04}. We explore to what extent type information can streamline the reversible semantics for interacting processes. Our discovery is that known (run-time) mechanisms used to support asynchronous (queue-based) semantics and  monitoring can also support  reversible protocol actions.

A key technical device in formalizing reversible semantics are \emph{memories}: these are run-time constructs which enable one to revert actions. Memories are the bulk of a reversible model; their maintenance requires care, as demonstrated by Tiezzi and Yoshida~\cite{DBLP:journals/jlp/TiezziY15}, who adapt known reversible semantics~\cite{DanosK04,LaneseMS10} into the session typed setting. 
In this work, we explore a different approach: 
we use monitors as memories.
%Rather than adapting known approaches to reversibility to the session typed case, w
We investigate to what extent queue-based semantics for session types can support reversibility. The key idea is simple: we use the type-checking component of queue-based semantics (i.e., the fact that session types enable process reductions) to support reversible process actions. 
Our approach concerns directly the reduction semantics for sessions, so we illustrate it via approximate reduction rules, which omit unimportant notational details. 
Consider the reduction rule for session communication, enhanced with 
session types  
and message 
queues, in the style of~\cite{Kouzapas09,DBLP:conf/ecoop/HuKPYH10,DBLP:conf/forte/KouzapasYH11}: 
\begin{equation}
\out{s}\msg{v}.P \parallel 
\bar{s}\queue{\send U.S_1 \,\cdot\, \tilde{h}_1} 
\parallel
s(x).Q
\parallel
s\queue{\receive U.S_2 \,\cdot\, \tilde{h}_2} 
\longrightarrow 
P \parallel 
\bar{s}\queue{S_1 \,\cdot\, \tilde{h}_1} 
\parallel
Q
\parallel
s\queue{S_2 \,\cdot\, \tilde{h}_2,v} 
\label{eq:introrule1}
\end{equation}
In \eqref{eq:introrule1},  
%Above, 
processes 
$\out{s}\msg{v}.P$ and $s(x).Q$ denote output and input along 
 \emph{session endpoints}  $\bar{s}$ and $s$, respectively. 
Notice that $\bar{s}$ and $s$ are \emph{dual} endpoints.
 Given an endpoint $s$, 
 process
 $s\queue{S\,\cdot\, \tilde{h}}$ is a   \emph{monitor}, where $S$ and $\tilde{h}$ are the session type  and message queue for~$s$, respectively. 
In the approach of~\cite{Kouzapas09,DBLP:conf/ecoop/HuKPYH10,DBLP:conf/forte/KouzapasYH11},  
session types enable communication actions: a synchronization can only occur if the actions (in the processes) correspond to the intended protocols (in the monitor types). After synchronization, portions of both processes and monitor types are consumed. Our approach consists in keeping, rather than consuming, these monitor types. For this to work, we need to distinguish the part of the protocol that has been already executed (its past), from the protocol that still needs to execute (its present). 
We thus introduce session types with \emph{present and past}: the type $S \past  T$ says that actions abstracted by $S$ are past protocol actions, whereas actions in $T$ are present steps. We may refine \eqref{eq:introrule1} as follows:
\begin{equation}
\out{s}\msg{v}.P \parallel 
\bar{s}\queue{T \past \send U.S_1 \,\cdot\, \tilde{h}_1} 
\parallel
s(x).Q
\parallel
s\queue{T' \past\receive U.S_2 \,\cdot\, \tilde{h}_2} 
\fw 
P \parallel 
\bar{s}\queue{T.\send U \past S_1 \,\cdot\, \tilde{h}_1} 
\parallel
Q
\parallel
s\queue{T'.\receive U \past S_2 \,\cdot\, \tilde{h}_2,v} 
\label{eq:introrule2}
\end{equation}
This is a \emph{forward} reduction rule.
%We call \emph{monitors}:
Monitors 
$\bar{s}\queue{T \past \send U.S_1 \,\cdot\, \tilde{h}_1}$ and $s\queue{T' \past\receive U.S_2 \,\cdot\, \tilde{h}_2}$ 
use type-checking to enable forward and backward computations; they may also implement asynchronous communication.
Observe that we use the cursor $\past$ to preserve 
output and input protocol actions  (noted $\send U$ and $\receive U$, respectively). 
Based on \eqref{eq:introrule2}, 
we may state a corresponding \emph{backward} reduction rule,
which reverts the intra-session synchronization at the level of processes, types, and message queues:
\begin{equation}
P \parallel 
\bar{s}\queue{T_1.\send U \past S_1 \,\cdot\, \tilde{h}_1} 
\parallel
Q
\parallel
s\queue{T_2.\receive U \past S_2 \,\cdot\, \tilde{h}_2,v} 
\bk
\out{s}\msg{v}.P \parallel 
\bar{s}\queue{T_1 \past \send U.S_1 \,\cdot\, \tilde{h}_1} 
\parallel
s(x).Q
\parallel
s\queue{T_2 \past\receive U.S_2 \,\cdot\, \tilde{h}_2} 
\label{eq:introrule3}
\end{equation}
%We stress that forward and backward reduction rules \eqref{eq:introrule2} and \eqref{eq:introrule3} are only an approximate illustration of our approach. 
Our main technical contribution is a core framework for session communication and reversibility whose monitored semantics
follows the spirit  of  rules \eqref{eq:introrule2} and \eqref{eq:introrule3}.
In our framework, session processes occur within \emph{configurations}, which add monitors and state for endpoints: while state conveniently implements substitutions, monitors handle both communication and reversibility.
Reduction is defined for configurations following rules \eqref{eq:introrule2} and \eqref{eq:introrule3}.
We support session establishment and  the consistent use of sent values and open variables in the state (cf. $v$ and $x$ in \eqref{eq:introrule2} and \eqref{eq:introrule3}). Our semantics enjoys the so-called ``loop lemma'', which offers a basic consistency guarantee for the interplay of forward and backward actions. 

In our opinion, the use of monitors with type-checking for reversible semantics is an observation that has 
at least two significant implications. %There are at least two main implications. 
First, it is encouraging to discover that monitor-based semantics with type-checking---introduced 
%the use of session types within (reduction) semantics, 
in~\cite{Kouzapas09,DBLP:conf/ecoop/HuKPYH10,DBLP:conf/forte/KouzapasYH11}
for asynchronous communications with events and used in~\cite{digiusto:hal-01093090,DBLP:journals/soca/CoppoDV15} to define run-time adaptation---may also inform the semantics of reversible protocols. 
%(The connection between adaptation and reversibility is somewhat expected: one may see reversibility as a particular form of adaptation.)
Monitors have also been used for security purposes~\cite{DBLP:journals/corr/abs-1108-4465,DBLP:journals/corr/CastellaniDP14} and, quite recently, for assigning blame to deviant session processes~\cite{LGP2016}.
Therefore, a monitor-based semantics encompasses an array of seemingly distinct concerns in structured communications.
Second, we see our developments as a first step towards validation techniques for communication and reversibility based on run-time verification.
Session frameworks with run-time verification have been developed in, e.g.,~\cite{DBLP:conf/forte/BocchiCDHY13,DBLP:journals/fmsd/DemangeonHHNY15}.
As these works do not support reversibility, our work may suggest enhancements for their dynamic verification capabilities.

%\noindent\textbf{Related Work}
%Works by Tiezzi and Yoshida~\cite{DBLP:journals/jlp/TiezziY15} (reversible sessions with heavy memories and without monitored semantics), and by Barbanera et al (reversible contracts ~\cite{BDLL15}).

%\paragraph{Organization.}
%In the following section we present 

\section{Syntax and Semantics}
\begin{figure}[t]
\begin{align*}
k, k' \sdef \quad& s, \dual{s} \sepr x, y \quad\quad\quad u, u' \sdef \quad a,b \sepr x,y  \quad\quad\quad n, m \sdef \quad a,b \sepr s,\dual{s}\\
M,N \sdef \quad& \nil \sepr  
%\widetilde{s}:(P \sep \store) 
\conf{\envs}{P}{\store}
\sepr 
%s\queue{H \sep \set{x}\sep \widetilde{u}}
\moni{a}{H}{\set{x}}{\widetilde{u}}
\sepr   \fresh{n} M \sepr M \parallel N \\ 
P,Q \sdef \quad& u(x:S).P \sepr \dual{u}(x:S).P \sepr\out{k}\msg{v}.P \sepr k(x).P \sepr 
%\keyword{close}(k).P  \ou P \parallel Q \ou 
\fresh{a} P \sepr\nil \\
S,T \sdef \quad&\keyword{end} \sepr \send U.S \sepr\receive U.S \\
H,K \sdef \quad & \past S \sepr S\past \sepr S \past T
\end{align*}
%\vspace{-4mm}
\caption{Syntax of Configurations, Processes, and Session Types.}
\label{fig:syn}
\end{figure}

%\paragraph{Conventions.} 
In this section we present our framework of session processes with monitored semantics and reversibility. 
We assume the following denumerable infinite mutually disjoint sets:
%the set $\names$ of \textit{names}, 
the set $\sessions$ of \textit{session names} (or \emph{endpoints}), the set $\channels$ of \textit{channels} and the set of \textit{variables} $\vars$. 
The set $\names = \sessions \cup \channels$ is called the set of \textit{names}. We assume a total bijection over $\names$, noted $\dual{\cdot}$, relating names with their duals such that $\dual{n} \neq n$ and $\dual{\dual{n}} = n$, for any name~$n$. We let $a,b$ to range over $\channels$; $s,r$ (and their duals) to range over $\sessions$; $m,n$ to range over $\names$ and $x,y$ to range over $\vars$. 
We use $\widetilde{o}$ to denote a finite sequence of objects (names, sessions, variables) $o_1, o_2, \ldots, o_n$, which we sometimes treat as a set or as an ordered list.  %(CHECK/ADJUST THIS).
We write \envs, $\envs'$ to range over finite, possibly empty sequences of  session names.

\paragraph{Syntax.}
%\noindent\textbf{Syntax.}
The syntax of 
configurations $M, N$, processes $P, Q$, and session types $S, T$
%our reversible session calculus 
is given in Fig.~\ref{fig:syn}. 
The syntax of  $M$ includes the empty configuration $\nil$, 
the \textit{running} process 
%$\set{s}:(P \sep \store)$, 
$\conf{\envs}{P}{\store}$, 
a monitor 
%$s\queue{T \sep \set{x} \sep \set{u}}$, 
$\moni{s}{S}{\set{x}}{\set{u}}$,
the name restriction $\fresh{n}M$, and parallel composition $M\parallel N$.
A running process $\conf{\envs}{P}{\store}$ is {univocally} identified by $\envs$, 
 the sequence of session endpoints occurring  in $P$. The local store $\store$  is a list of pairs of the form $\{x,\set{v}\}$ (see Def. \ref{d:store}). A monitor  
%$s\queue{S\sep\set{x}\sep \set{u}}$ %(respectively $\dual{k}\queue{S;\set{x};u}$) 
$\moni{s}{S}{\set{x}}{\set{u}}$
is identified by the session name $s$, % (resp. $\dual{k}$), 
contains its session type $S$ (see below), a list $\set{x}$ containing all the variables used by the process, and a list of names $\set{u}$ that the process has used in the session. These two lists will be useful to rebuild prefixes.  
A monitor type $T$ describes the behavior of its associated session. 
The syntax of types assumes a set of basic \textit{sorts} ($\basic{bool}$, $\basic{int}$, $\ldots$), ranged over $U$. 
We also assume  $\prim$ as the set of all possible values 
belonging to basic sorts; this way, $\val = \names \cup \prim$ is the set of values that processes can exchange.
We use $v,w$ (and their decorated versions) to range over $\val$.
The type $\send U.S$ (resp. $\receive U.S$) indicates that the owner of the monitor may send (resp. receive) a value of type $T$ and proceed with the behavior prescribed by $S$. 

Types $\send U.S$ and $\receive U.S$ are standard in session types disciplines. A novelty in our work is the (run-time) type $S_1\past S_2$:
it indicates that  $S_1$ is the past (already executed) behavior of the associated session,  while $S_2$ represents the 
present behavior (yet to be executed). 
That is, the separator $\past$ is used as a cursor in a type; it is 
inspired by the separator used in~\cite{CardelliL11} to remember the past of sequential CCS processes.
These session types with \emph{present and past} occur only at run-time; the
%The separator $\past$ only occurs at top-level; types such as $(S_1 \past S_2) \past T$ is excluded.
intent is that each time that the process performs a forward computation the cursor will be moved forward by one action; it will be moved backwards by one action as result of a reversible action. 
%The variable list indicates the variables used in the session: if the session type is $S_1\past S_2$, then the variable list contains all the variables used by the process obeying type $S_1$. 

The syntax of processes follows standard lines:
we consider 
the idle process $\nil$, %the session closing process $\keyword{close}(k)$, % (resp. $\keyword{close}(\dual k)$), 
prefixes for session establishment (noted $u(x:S).P$ and $\dual{u}(x:S).Q$, where $S$ is a session type), and 
prefixes for intra-session communication (noted $k(x).P$ and $k\msg{v}.P$).
%, and the parallel composition $P\parallel Q$.
We write $\procs$ and $\confs$ to indicate the set of processes and configurations, resp. We call \textit{agent} an element of the set $\agents = \confs \cup \procs$. We let $P,Q$ (and their decorated versions)  to range over $\procs$; also, we use $L,M,N$ to range over $\confs$ and $A, B, C$ to range over $\agents$. 

Before formally presenting the operational semantics, 
we give some intuitions on the information carried by monitors, in particular the variable lists.
%\begin{example}
Consider the following configuration, with $s\in \envs$:
$$\conf{\envs}{P}{\store} \parallel \moni{s}{S.\receive U\past T}{\set{x},x}{\set{u},k}$$ 
By inspecting the type before the cursor $\past$, 
we know the last action of the process, before becoming $P$, was an input; 
also, by the additional information in the variable and name lists, we know that this input action was of the form $k(x)$.
That is, the shape of the process right before the input action was $k(x).P$. 
%\end{example}

%\noindent\textbf{Operational Semantics.}
\paragraph{Operational Semantics.}
The operational semantics of our reversible calculus is defined via a reduction relation $\red$, which is a binary relation over
configurations $\red \subset \confs \times \confs$, and a structural congruence relation $\equiv$, which is a binary relation over processes and configurations $\equiv \subset \procs^{2}\cup \agents^{2}$.

\begin{figure}
{\small
\begin{mathpar}
	\inferrule*[left=(E.ParC)]{}
	{A \parallel B \equiv B \parallel A}
	\and
	\inferrule*[left=(E.ParA)]{}
	{A \parallel (B \parallel C) \equiv (A \parallel B) \parallel C}
	\and
	\inferrule*[left=(E.NilM)]{}
	{A \parallel \nil \equiv A}
	\and
	\inferrule*[left=(E.NewN)]{}
	{\fresh{n}{\nil} \equiv \nil }
	\and
	\inferrule*[left=(E.NewC)]{}
	{\fresh{n}{\fresh{m}{A}} \equiv \fresh{m}{\fresh{n}{A}}}	
	\and
	\inferrule*[left=(E.NewP)]{}
	{(\fresh{n}{A}) \parallel B \equiv \fresh{n}{(A \parallel B)}}
	\and
	\inferrule*[left=(E.$\alpha$)]{}
	{A =_{\alpha} B \implies A \equiv B}
\end{mathpar}
}
\vspace{-6mm}
\caption{Structural congruence}
\label{fig:str}
\end{figure}
\begin{definition}[Contexts]
\textit{Configuration contexts}, also called evaluation contexts, are configurations with one hole ``$\hole$'' defined by the following grammar: $\ctx{E} \sdef \hole \sepr( M \parallel \ctx{E}) \sepr\fresh{n}\ctx{E}$.
General contexts $\ctx{C}$ are processes or configurations with one hole $\hole$'', and are obtained from processes or configurations by replacing one occurrence of  $\nil$ (either as process or as configuration) with $\hole$.
\end{definition}
A congruence on processes and configurations is an equivalence relation $\rel$ that is closed under general contexts: $P \,\rel\,  Q \Longrightarrow \ctx{C}[P]\,\rel\,  \ctx{C}[Q] $
and $M\rel N \Longrightarrow \ctx{C}[M]\,\rel\,  \ctx{C}[N]$.
The relation $\equiv$ is defined as the smallest congruence, on processes and configurations, that satisfies rules in Figure~\ref{fig:str}. In defining the rules  
we adopt Barendregt's Variable Convention: If terms $t_1, \ldots, t_n$ occur in a certain context (e.g. definition, proof), then in these terms all bound identifiers and variables are chosen to be different from the free ones. This is why in Rule $(\textsc{E.NewP})$ there is no check on free names.

A binary relation $\rel$ on closed configurations is \textit{evaluation-closed} if it satisfies the inference rules:
%\small
\begin{mathpar}
\inferrule*[left=(Ctx)]{M\,\rel\, N}{\ctx{E}[M] \,\rel\,  \ctx{E}[N]} \and  \inferrule*[left=(Eqv)]{M\equiv M' \and M'\,\rel\,  N' \and N'\equiv N}{M \,\rel\,  N} 
\end{mathpar}
The reduction relation $\red$ is defined as the union of two relations, the \textit{forward} reduction relation $\fw$ and the 
backward reduction relation $\bk$: $\red = \fw \cup \bk$. Relations $\fw$ and $\bk$ are the smallest evaluation-closed relations satisfying the rules in Figure~\ref{fig:sem}.
Before commenting them we need some definitions in place:
\begin{definition}[Dual type]
The dual type of a type $S$, indicated as $\dual{S}$, is inductively defined as follows:
%\begin{align*}
$$\dual{\send{U}.S} = \receive{U}.\dual{S} \quad \quad \dual{\receive{U}.S} = \send{U}.\dual{S} \quad \quad \dual{\keyword{end}} = \keyword{end}$$
%\end{align*}
\end{definition}
Sometimes we will write $\mathsf{dual}(S_1,S_2)$ to indicate $S_1 = \dual{S_2}$.

\begin{remark}[Store and Explicit Substitution]
One of the main challenges in defining reversible 
semantics for processes is how to treat substitutions, since in general a substitution is not a bijective function. 
There are at least two possibilities. First, one may create a copy of a process before applying a substitution and then replace the process with its copy when reverting the substitution; this is the strategy developed in~\cite{LaneseMS10}.
Second, one may 
use a store and a mechanism with explicit  substitution, following~\cite{LienhardtLMS12,DBLP:conf/lics/CristescuKV13}. The first technique creates a memory each time a value is substituted; 
here we implement the second technique, in which one just has to remember the pair variable/value for each substitution. 
\end{remark}
We now formally define the store $\sigma$ and the operations on it. 

\begin{definition}%[Local Store and Operations]
\label{d:store}
A local store $\store$ is a mapping from variables to an ordered list of values. 
Given a store $\store$, a variable $x$, and a value $v$, we define 
the \textit{update} $\store[x \mapsto v]$ and \textit{reverse update} $\store \rup x$ as follows:
\begin{align*}
	 \store[x \mapsto v] =  \left\{ \begin{array}{rl} \store \cup \{x, v\} & \text{ if } x \not \in dom(\store)  \\
	 \store_1 \cup \{x,\widetilde{v}\cdot v\}  & \text{ if } \store = \store_1 \cup \{x, \widetilde{v} \} 	\end{array} \right .
	 & \qquad 
	 \store \rup x =  \left\{ \begin{array}{rl} \store_1  & \text{ if } \store= \store_1 \cup \{x, v \}  \\
	 \store_1 \cup \{x, \widetilde{v}\}  & \text{ if } \store = \store_1 \cup \{x, \widetilde{v}\cdot v\} 	\end{array} \right .
 \end{align*}
The \textit{evaluation} of name $n$ under a local store $\sigma$, written  $\myeval{n}{\sigma}$, is the value $v$ if $\{n,v\}\in \store$ or $\{n,\widetilde{v}\cdot v\} \in \store$; otherwise, it is $n$ itself. 
 %That is, the evaluation of a name $n$ is the value at the top of the list that correspond to it, if it exists; otherwise, it is the name itself (meaning that the given name is not a variable). 
 Let us stress the fact that our store maintains a correspondence between variables and \textit{lists} of values in order to enable reversibility. The list represents at a given time the assignment history of the variable to which it corresponds, with the actual value of the variable being at the top of the list.
 If $\myeval{n}{\store} = n$ then $n$ is not a variable.
% while the \textit{reverse update} $\store \rup x $ is defined as follows:
%\[
%	 \store \rup x =  \left\{ \begin{array}{rl} \store_1  & \text{ if } \store= \store_1 \cup \{x, v \}  \\
%	 \store_1 \cup \{x, \widetilde{v}\}  & \text{ if } \store = \store_1 \cup \{x, \tilde{v}\cdot v\} 	\end{array} \right .
% \]
% 
%$\sigma:\vars \rightarrow \widetilde \names$
\end{definition}

\begin{figure}[t!]
%\small
\begin{mathpar}
\inferrule[(Open)]{
\mathsf{dual}(S,T) \and 
\dual{s} \not\in\envs 
\and s \not\in\envs'  \and \myeval{u}{\sigma_1} =  {\myeval{u'}{\sigma_2}
%=a
}
%\qquad \set{x} = fn(P) \qquad \set{y} = fn(Q)
}
{
\conf{\envs}{\dual{u}(x:S).P}{\store_1}
\,\parallel\, 
\conf{\envs'}{u'(y:T).Q}{\store_2}
\\
\fw 
\\ (\new s,\dual{s}).\,\big(
\conf{\envs,\,\dual{s}}{P}{\store_1[x\mapsto \dual{s}]}  
\,\parallel\,
\moni{\dual{s}}{\past  S}{x}{\dual{u}}  
\,\parallel\, 
\conf{\envs',\,s}{Q}{\store_2[y\mapsto s]} 
 \,\parallel\, 
 \moni{s}{\past   T}{y}{u'} \big)
 } 
\and
\inferrule[(Open$\un$)]
{\mathsf{dual}(S,T)  \and \myeval{u}{\sigma_1} = {\myeval{u'}{\sigma_2}
%=a
} 
% \store_1=\store \cdot (x,\set{v_1}\cdot \dual{s}) \and \store'_1=\store \cdot (x,\set{v_1}) \and
%\store_2=\store \cdot (y,\set{v_2}\cdot s) \and \store'_2=\store \cdot (x,\set{v_2}) 
}
{
(\new s,\dual{s}).\, 
\big(
\conf{\envs,\,\dual{s}}{P}{\store_1}  
\,\parallel\, 
\moni{\dual{s}}{\past  S}{ x}{\dual{u}}  
\,\parallel\, 
\conf{\envs',\, s}{Q}{\store_2}
 \,\parallel\, 
\moni{s}{\past   T}{y}{u'}  
 \big)
\\
\bk
\\
\conf{\envs}{\dual{u}(x:S).P}{\store_1 \rup x} 
\,\parallel\, 
\conf{\envs'}{u'(y:T).Q}{\store_2 \rup y}
}
 \and
%\inferrule*[left=(Com)]{\eval{k}_{\sigma_1}  = s \and s\in \widetilde{s_1} \and \eval{\dual{k}}_{\sigma_2} = \dual{s} \and  \dual{s} \in \widetilde{s_2} \and \eval{y}_{\store_2} = v }
%{\widetilde{s_1}: (k(x).P ; \store_1)  \,\parallel\, \widetilde{s_2}: (\dual{k}\msg{y}.Q ; \store_2) \,\parallel\, s \queue{R_1.\past \receive T.S_1 ; \set{x}}  \,\parallel\, \dual{s}\queue{R_2.\past \send T.S_2 ; \set{y}} \fw  \\ 
%\widetilde{s_1}: (P;\store_1[x\mapsto \set{v}\cdot v])  \,\parallel\, \widetilde{s_2}: (Q;\store_2) \,\parallel\, s \queue{R_1. \receive T.\past S_1 ; \set{x},x}  \,\parallel\, \dual{s}\queue{R_2. \send T.\past S_2 ; \set{y},y} } \and
\inferrule[(Com)]{
\myeval{k}{\sigma_1}  = s \and s\in \envs \and \myeval{\dual{k}'}{\sigma_2} = \dual{s} \and  \dual{s} \in \envs' \and \myeval{y}{\store_2} = v }
{
\conf{\envs}{k(x).P}{\store_1}
\,\parallel\, 
\conf{\envs'}{\dual{k}'\msg{y}.Q}{\store_2}
\,\parallel\, 
\moni{s}{T.\past \receive U.S_1}{\set{x}}{\set{n}}  
\,\parallel\, 
\moni{\dual{s}}{T'.\past \send U.S_2}{\set{y}}{\set{m}} 
\\ \fw  \\ 
\conf{\envs}{P}{\store_1[x\mapsto  v]}  
\,\parallel\, 
\conf{\envs'}{Q}{\store_2}
\,\parallel\, 
\moni{s}{T. \receive U.\past S_1}{\set{x},x}{\set{n},k}  
\,\parallel\, 
\moni{\dual{s}}{T'. \send U.\past S_2}{\set{y},y}{\set{m},k'}
}
 \and
\inferrule[(Com$\un$)]
{\myeval{k}{\sigma_1}  = s \and s\in \envs \and \myeval{\dual{k}'}{\sigma_2} = \dual{s} \and  \dual{s} \in \envs' 
%\and \store_1=\store \cdot (x,\set{v}\cdot v) \and \store'_1=\store \cdot (x,\set{v})
} 
{
\conf{\envs}{P}{\store_1}
\,\parallel\, 
\conf{\envs'}{Q}{\store_2}
\,\parallel\, 
\moni{s}{T. \receive U.\past S_1}{\set{x}, x}{\set{n}, k}  
\,\parallel\, 
\moni{\dual{s}}{T'. \send U.\past S_2}{\set{y}, y}{\set{m}, k' } 
\\
\bk  \\
\conf{\envs}{k(x).P}{\store_1 \rup x}
\,\parallel\, 
\conf{\envs'}{\dual{k}\msg{y}.Q}{\store_2}
\,\parallel\, 
\moni{s}{T.\past  \receive U.S_1}{\set{x}}{\set{n}}  
\,\parallel\, 
\moni{\dual{s}}{T'. \past \send U. S_2}{\set{y}}{\set{m}}  
} 
%\and
%\inferrule*[left=(Par)]{\widetilde{k}:P \parallel M \red \widetilde{k'}:P' \parallel M'}
%{\widetilde{k_1}\cdot\widetilde{k}:(P \parallel Q )\parallel M \red \widetilde{k_1}\cdot\widetilde{k}':(P' \parallel Q ) \parallel M'}
\end{mathpar}
%\vspace{-4mm}
\caption{Operational Semantics: Forward ($\fw$) and Backward ($\bk$) Reduction Semantics.}
\label{fig:sem}
\end{figure}

The rules of the operational semantics are in Fig. \ref{fig:sem}. We briefly comment on them:
\begin{enumerate}[$\bullet$]
\item Rule \textsc{Open} is the forward rule for session establishment. It creates two fresh, dual endpoints and their associated monitors. 
Each monitor stores a session type; each store records the name of the fresh session endpoints and of the channel. The monitor records the name on which the session has started and the variable used by the process to refer to the endpoint. Establishing a new session requires type duality and that the two process refer to the same name (cf. condition $\myeval{u}{\sigma_1} = \myeval{u'}{\sigma_2}$).
%(SPEIGARE PREMESSA $\eval{u}_{\sigma_1} = \dual{\eval{\dual{u}}_{\sigma_2}}$ CHE NON MI E' CHIARA.)

\item Rule \textsc{Open}$\un$ is the exact opposite of Rule $\textsc{Open}$. In order to revert a session creation, the rule checks that session types are at their initial position (e.g., $\past S$ and $\past T$). Moreover, the variable and name lists should contain just one element each. The effect of the rule is to collect the two endpoint names and to eliminate the two associated monitors, while restoring the prefixes in the processes.

\item Rule \textsc{Com} describes intra-session communication. Two running processes can communicate if they refer to the same session. The sent value $v$ is obtained by evaluating the sent value $y$ under the sender store $\store_2$. 
The result of a communication is that the store of the receiver is updated with a new value, bound to the read variable: $\store_1[x\mapsto v]$. Also, both monitor types are moved one step forward, and both read and sent variables are put on the top of the list in their respective monitors; the same occurs for the session names. This way we keep information about the  {prefixes} (i.e., $k(x)$ and $\dual{k}'\msg{y}$).

\item Rule \textsc{Com}$\un$ undoes a communication: the sent value (along with the variable used to read it) is eliminated from the store of the receiving process. 
The variable list of the monitor keeps information on which variable to unbound: indeed, it is the variable at the top of this list the one that has to be eliminated, as we want to revert  the input of its associated value. Moreover, information contained in the name list of the monitors is used to recover output and input prefixes. Notice that the information about the kind of prefix to be built again (input or output) is given by the type of the monitor. A further consequence of undoing a communication is that the session types are moved one step backward.
\end{enumerate}

%TO DO: COME RICORDARSI DAL NESTING DI OPENS (SEQUENZA LINEARE DI NOMI $u_1, u_2, \ldots, u_n$)
\noindent
Our process framework satisfies the so-called \emph{loop lemma}, a property that gives us a basic guarantee of the consistency between 
forward and backward reductions. We require the following definition:

\begin{definition}[Initial and Reachable Configurations]
A configuration is \textit{initial} if there are no monitors and all the running processes have an empty store and are identified by $\emptyset$.
A configuration is \textit{reachable} if it can be derived using $\red$ from an initial configuration.
\end{definition}
An easy induction on the structure of terms provides us
with a kind of normal form for configurations:
\begin{lemma}[Normal Form]\label{lm:nf}
For any configuration $M$ we have that: 
$$M\equiv \new{\widetilde a}. \left( \; \prod_{i\in I} \conf{\envs_{\, i}}{P_i}{\sigma_i} \parallel \prod_{j\in J} \moni{s_j}{H_j}{\widetilde y_j}{\widetilde u_j} \right)
 $$
\end{lemma}
Notice that, by convention, we assume that   $\prod_{i\in I} A_i = \nil$ if  $I = \emptyset$.
Then we have:
\begin{lemma}[Loop Lemma]
For any reachable configuration $M,N$, we have $M \fw N$ $\iff$ $N\bk M$.
\end{lemma}
%\begin{proofout} 
\begin{proof}[Proof (Sketch)]
By induction on the derivation of $M\fw N$ for the if direction, and on the derivation of $N\bk M$ for the converse.
We will just show the forward case when the applied rule is $\textsc{Open}$; the other cases are similar. 
By Lemma~\ref{lm:nf} we have that: 
$$M\equiv \new{\widetilde a}. \left( \; \prod_{i\in I} \conf{\envs_{\, i}}{P_i}{\sigma_i} \parallel \prod_{j\in J} \moni{s_j}{H_j}{\widetilde y_j}{\widetilde u_j} \right)
 $$
and since rule $\textsc{Open}$ is applied, then there exist two indexes $w,z\in I$ such that 
$P_{w} = \dual{u}(x:S).Q_{w}$ and 
 $P_{z} =u'(y:T).Q_z$ with 
 $\mathsf{dual}(S,T)$, $\dual{s} \not\in\envs_w$, $s \not\in\envs_z$  and  
 $\myeval{u}{\sigma_1} =  \myeval{u'}{\sigma_2}$.  We have then:
\begin{align*}
M \fw  \new{\widetilde a, s,\dual{s}}.& \big( \; \prod_{i\in I'} \conf{\envs_{\, i}}{P_i}{\sigma_i} \parallel \prod_{j\in J} \moni{s_j}{H_j}{\widetilde y_j}{\widetilde u_j}  \parallel \conf{\envs_w, \, \dual s}{P_w}{\store_w[x\mapsto \dual s ]}  
 \parallel \conf{\envs_z,\,s}{P_z}{\store_z[y\mapsto s ]} \\
 & \quad \parallel  \moni{s_w}{\past T}{x}{\dual u} \parallel \moni{s_z}{\past S}{y}{u} \big)
 \end{align*}
 with $I' = I\setminus \{w,z\}$. 
 It is easy to see that by applying \textsc{Open}$\un$ we get back to $M$, as desired.
%\end{proofout}
\end{proof}

\noindent Another property that a reversible calculus should enjoy is the so-called \textit{square lemma}, which may be informally described as follows. 
Assume a configuration from which two reductions are possible: if these reductions are \textit{concurrent} then the order in which the two reductions are applied does not matter, and the same configuration is reached. This lemma therefore relies on the definition of concurrent transitions. In our setting, thanks to the information carried by the monitors and to linearity of sessions,  we may decree that two reductions are concurrent if they operate on different sessions/channels (service names). 
One may instrument the reduction semantics  $\red$ with a label $\lambda$ containing the endpoints used by the rule and the service name (if any):
$M\xrightarrow{\lambda} N$. 
Then,  reductions $M\xrightarrow{\lambda_1} N$  and $M\xrightarrow{\lambda_2} N$ are  concurrent if $\lambda_1 \cap \lambda_2  = \emptyset$.

Using the square lemma one may then show that the reversible semantics is \emph{causally consistent}, i.e., that 
an action can be reverted
only after all the actions causally depending on it have already been reverted.
We leave for future work establishing these further results for our framework.

\section{Concluding Remarks and Future Work}
We have proposed a fresh approach to reversible semantics for session processes: it builds upon a style of process semantics in which monitors (which include session types) enable process reductions. Even if this style of process semantics is not new---it was introduced in~\cite{Kouzapas09} and later used in~\cite{DBLP:conf/ecoop/HuKPYH10,DBLP:conf/forte/KouzapasYH11,digiusto:hal-01093090,DBLP:journals/soca/CoppoDV15}---to our knowledge this is the first time that this formulation is used to support reversibility. 

We rely on monitors which contain types that describe past and future structured interactions; these types offer a natural 
form of memories for supporting forward and backward semantics. We motivated our approach by introducing a simple process framework with session establishment and communication; extensions with other usual session constructs (labeled choice and recursion) are straightforward. 
To highlight the simplicity of our approach, we have considered binary session types. We believe that our approach scales to account for multiparty structured communications; in such a setting, monitors would be generated after multiparty session establishment, and would be equipped with local projections of global types, as in~\cite{DBLP:conf/forte/BocchiCDHY13,DBLP:journals/corr/CastellaniDP14}. 
A multiparty, asynchronous semantics may need to consider forms of \emph{coordinated} reversibility among different partners; 
we plan to address these challenges in future work.

Most models of reversible  processes (cf.~\cite{DanosK04}) do not consider (behavioral) types, and so their reversible semantics must account for arbitrary  forms of concurrent behavior. 
In reversing the untyped $\pi$-calculus, substitutions and scope extrusion are known to be challenging issues~\cite{DBLP:conf/lics/CristescuKV13}.
Reversing session processes is a seemingly simpler problem, as behavior is disciplined by types: once a session is established, concurrency interactions proceed in a deterministic, confluent manner. Also, in session $\pi$-calculi scope extrusion is limited. 
To our knowledge, the work~\cite{DBLP:journals/jlp/TiezziY15} is the first to address reversibility for 
a synchronous $\pi$-calculus with binary session types.
A key difference between our work and~\cite{DBLP:journals/jlp/TiezziY15} is the role that session types play in the reversible semantics. 
We have used session types to define forward and backward semantics for session processes; in contrast, the reversible semantics in~\cite{DBLP:journals/jlp/TiezziY15} establishes key results for reversibility (e.g., the square lemma and causal consistency) using an untyped reduction semantics. %(Applications of type checking to memories are left for future work.) 
Hence, although the reversible session $\pi$-calculus in~\cite{DBLP:journals/jlp/TiezziY15} is shown to be typable using standard binary session types, the influence of types on the reversible semantics of~\cite{DBLP:journals/jlp/TiezziY15} is indirect  at best.  

As further topics for future work, we plan to establish the precise savings involved from moving from (i)~an untyped reversible semantics to (ii)~a monitored reversible semantics with types, as proposed here. 
We plan to compare the (untyped) reversible higher-order processes in~\cite{LaneseMS10}
and the core higher-order session calculus in~\cite{KPY2016}, which may precisely encode the first-order session $\pi$-calculus.
Another direction  concerns \emph{controlled reversibility}~\cite{LaneseMSS11}.
Some recent approaches have proposed types with controlled roll-back and explicit checkpoints among parties~\cite{BDLL15}. We believe that by adding explicit rollback into (run-time) type information, we could achieve intuitive mechanisms for controlled reversibility. 
%\vspace{-4mm}
\paragraph{Acknowledgments.} We are grateful to Dimitris Kouzapas for useful exchanges.
We would also like to thank the anonymous reviewers for their suggestions, which were helpful to improve the presentation.
P\'erez is also affiliated to the NOVA Laboratory for Computer Science and Informatics (NOVA LINCS - PEst/UID/CEC/04516/2013), Universidade Nova de Lisboa, Portugal.
%\vspace{-4mm}
\bibliographystyle{eptcs}
\bibliography{rsessions}
\end{document}